\documentclass[onecolumn,twoside]{IEEEtran}%
\usepackage{amsfonts}
\usepackage{amsmath,amscd}
\usepackage{amssymb}
\usepackage{graphicx}%
\usepackage{color}
\usepackage{verbatim}
\usepackage{hyperref}
\usepackage{tikz}	
\usetikzlibrary{backgrounds,fit,decorations.pathreplacing}  
\usetikzlibrary{automata} 
\usepackage{lipsum}

\definecolor{darkred}  {rgb}{0.5,0,0}
\definecolor{darkblue} {rgb}{0,0,0.5}
\definecolor{darkgreen}{rgb}{0,0.5,0}
\hypersetup{
  pdftitle = {The Split-Source Cost of Three-Party Correlations},
  pdfauthor = {Eric Chitambar, Min-Hsiu Hsieh, Andreas Winter},
  colorlinks = true,
  urlcolor  = blue,         
  linkcolor = darkblue,     
  citecolor = darkgreen,    
  filecolor = darkred       
}

\usepackage{framed}
{\begin{leftbar}\underline{\textbf{Internal Notes}}:\begin{quotation}}%
{\end{quotation}\end{leftbar}}

\newtheorem{theorem}{Theorem}

\newtheorem{definition}[theorem]{Definition}

\newtheorem{lemma}[theorem]{Lemma}

\newtheorem{proposition}[theorem]{Proposition}
\newtheorem{remark}[theorem]{Remark}

\newenvironment{proof}[1][Proof]{\noindent\textbf{#1.} }{\hfill\IEEEQED}

\newcommand{\be}{{\mathbf e}}

\newcommand{\bx}{{\mathbf x}}
\newcommand{\by}{{\mathbf y}}

\def\cU{{\cal U}}
\def\cV{{\cal V}}        \def\cW{{\cal W}}
\def\cX{{\cal X}}
\def\cY{{\cal Y}}        \def\cZ{{\cal Z}}

\def\0{{\mathbf{0}}}
\def\1{{\mathbf{1}}}
\def\2{{\mathbf{2}}}
\def\3{{\mathbf{3}}}
\def\4{{\mathbf{4}}}
\def\5{{\mathbf{5}}}
\def\6{{\mathbf{6}}}

\def\7{{\mathbf{7}}}
\def\8{{\mathbf{8}}}
\def\9{{\mathbf{9}}}

\def\be{\begin{equation}}
\def\ee{\end{equation}}
\def\bea{\begin{eqnarray}}
\def\eea{\end{eqnarray}}

\newcommand{\mbf}[1]{\mathbf{#1}}
\newcommand{\mc}[1]{\mathcal{#1}}

\newcommand{\qed}{{\hfill\IEEEQED}}

\begin{document}

\title{The Private and Public Correlation Cost \protect\\ of Three Random Variables with Collaboration}
\author{Eric Chitambar,$^1$ Min-Hsiu Hsieh,$^2$ and Andreas Winter\,$^3$ 
\\[4mm]  
\textit{$^1$ Department of Physics and Astronomy, Southern Illinois University,}\\ 
\textit{Carbondale, Illinois 62901, USA}\\[1mm]
\textit{$^2$ Centre for Quantum Computation \& Intelligent Systems (QCIS),}\\
\textit{Faculty of Engineering and Information Technology (FEIT),}\\
\textit{University of Technology Sydney (UTS), NSW 2007, Australia}\\[1mm]
\textit{$^3$ ICREA \&{} F\'isica Te\`orica: Informaci\'o i Fen\`omens Qu\`antics} \\
\textit{Universitat Aut\`onoma de Barcelona, ES-08193 Bellaterra (Barcelona), Spain}
\\[4mm]
(12 September 2014)}


\maketitle

\begin{abstract}
In this paper we consider the problem of generating arbitrary three-party 
correlations from a combination of public and secret correlations.  
Two parties -- called Alice and Bob -- share perfectly correlated bits that 
are secret from a collaborating third party, Charlie.  At the same time, 
all three parties have access to a separate source of correlated bits, 
and their goal is to convert these two resources into multiple copies of 
some given tripartite distribution $P_{XYZ}$.  We obtain a single-letter 
characterization of the trade-off between public and private bits that 
are needed to achieve this task.  The rate of private bits is shown to 
generalize Wyner's classic notion of \textit{common information} held 
between a pair of random variables.  The problem we consider is also 
closely related to the task of \textit{secrecy formation} in which $P_{XYZ}$ 
is generated using public communication and local randomness but with 
Charlie functioning as an adversary instead of a collaborator.  
We describe in detail the differences between the collaborative and 
adversarial scenarios.
\end{abstract}

\section{Introduction}
\label{sec:intro}
Three-party correlations are central objects of interest in the discussion of 
public key agreement \cite{Maurer-1993a, Ahlswede-1993a}.  
Two parties (Alice and Bob) have access to some source which generates multiple 
copies of three random variables $XYZ$.  When $n$ copies are generated, Alice 
sees $X^n$, Bob sees $Y^n$, and a third party (Charlie) sees $Z^n$.  
In the standard key agreement scenario, Charlie is viewed as untrustworthy 
eavesdropper and Alice and Bob wish to extract perfectly shared randomness 
from $X^nY^nZ^n$ using local randomness and public communication (LOPC).  
The security constraint is that at the end of this protocol,
Charlie should be almost completely uncorrelated from Alice and Bob's shared randomness.     
 
However, in many scenarios it may not be appropriate to assume that Charlie 
is a malicious eavesdropper.  In fact, Charlie may actually be a helper or 
collaborator to Alice and Bob in their pursuit of obtaining private randomness 
from $X^nY^nZ^n$.  For instance, one might imagine that Charlie represents 
some centralized hub that wishes to establish a secure link between two of 
its users.   
Distillation problems of this sort have been studied in Ref. \cite{Csiszar-2000a};
see also Refs. \cite{Gregoratti-2003,Smolin-2005,Winter-2007} for quantum analogues.

This paper considers the reverse of the scenario just described.  
Instead of asking how much secret key can be distilled from $XYZ$ using 
LOPC, we ask how much secret key is needed to build $XYZ$ using LOPC.  
This kind of problem has been studied much less, but goes back all
the way to Wyner \cite{Wyner-1975a}, and has received much more attention
only in the last decade or so, from the \emph{Reverse Shannon Theorem} \cite{Bennett-2002}
(and its quantum generalization \cite{Bennett-2014}), more generally
to so-called coordination problems \cite{Cuff-2008,Cuff-PhD-2009}.

Whether Charlie is an adversary or collaborator greatly changes the nature of the problem,
as we shall see. 
First consider when Charlie is a collaborator.  Alice and Bob initially share perfect randomness that is secret from Charlie, and using LOPC, they generate public communication $U$ and variables $\hat{X}^n\hat{Y}^n$.  However, in the spirit of collaboration, the public communication which they generate should also be usable by Charlie to generate $\hat{Z}^n$ so that ultimately $\hat{X}^n\hat{Y}^n\hat{Z}^n\approx X^nY^nZ^n$.  We will refer to this as the \textit{collaborative model} for generating $XYZ$.  In a particular protocol, there will exist some trade-off between the amount of secret randomness Alice and Bob initially share versus the amount of public communication used to build $\hat{X}^n\hat{Y}^n\hat{Z}^n$.  The main contribution of this paper is a single-letter characterization of this trade-off (Theorem \ref{Thm:Main}).

On the other hand, when Charlie is an adversary, some care is needed to properly quantify the correlation costs of $XYZ$.  This is because here Alice and Bob really only care about generating the marginal $XY$ since, after all, Charlie is an adversary.  Nevertheless, $Z$ may contain some information about $XY$, and this should be somehow captured in the total cost for $XYZ$.  In light of these considerations, Renner and Wolf have proposed the following notion of secrecy formation \cite{Renner-2003a}.  
Starting from a source of pre-shared secret bits, Alice and Bob perform LOPC to generate three random variables $\hat{X}^n\hat{Y}^nU$, where again $U$ describes the public communication conducted during the protocol.  With $X^nY^nZ^n$ being the target distribution, the goal is for $\hat{X}^n\hat{Y}^n U$ to be approximately equivalent to some joint random variables of the form $X^nY^n\overline{Z}$, where $\overline{Z}$ is obtained by processing $Z^n$.  This latter condition means that Charlie could simulate the entire communication Alice and Bob use to produce $\hat{X}^n\hat{Y}^n$ from his part $Z^n$.  In the words of Renner and Wolf, this ``formalizes the fact that the protocol communication $U$ observed by [Charlie] does not give him more information than $Z^n$.''  We will refer to this as the \textit{adversarial model} for generating $XYZ$.  

In subsequent work, Horodecki \textit{et al.} \cite{Horodecki-2005c} discussed the
hypothesis that the minimum rate of secret bits for generating $XYZ$ in the adversarial sense 
is given by a quantity known as the intrinsic information \cite{Renner-2003a}.  
If this were true, then optimal secrecy formation could alternatively be 
obtained by an asymptotic preparation of randomly chosen private correlations 
(i.e. distributions $P_{XYZ}$ in which Charlie is completely uncorrelated from Alice and Bob).  
However, Horodecki \emph{et al.} considered this hypothesis to be likely false,
and indeed it was later shown to be so by one of us, with a single-letter secret 
key cost formula being derived \cite{Winter-2005a}.  
In the present paper, we revisit the precise trade-off between public and secret 
correlation costs computed in Ref. \cite{Winter-2005a} for the adversarial model,
in particular the direct (achievability) part of the main result of that paper.
Note that the direct part is also implicitly discussed in \cite{Horodecki-2005c},
since the formula of the secrecy cost is a convex hull over certain decompositions
of the Alice-Bob distribution, decompositions characterized by Wyner's \emph{common information} \cite{Wyner-1975a}.
The converse (lower bound) presented in \cite[Prop.~1]{Horodecki-2005c} 
however was incomplete as it assumed a property known as ``asymptotic continuity'' 
(cf.~\cite{Horodecki-2009a})
of the common information, which was shown to be false in Ref. \cite{Winter-2005a};
see also Witsenhausen \cite{Witsenhausen-1976a}. The correct optimality
proof required a much more complex argument \cite{Winter-2005a}.

Both the collaborative and adversarial models generalize Wyner's notion of common information \cite{Wyner-1975a}.  In Wyner's scenario, Alice and Bob simply want to produce $X^nY^n$ using pre-shared randomness with no additional communication.  The minimum amount of randomness per copy needed to approximately simulate $X^nY^n$ is what Wyner identifies as the common information held between $X$ and $Y$.  This can be seen as a special case of the three-party problems when $Z$ ranges over just a single value and the public communication rate is zero.

Interest in secrecy formation is largely inspired by the analogous notion of \textit{entanglement formation} when dealing with quantum systems and quantum information.  The entanglement cost of a quantum state is defined to be the asymptotic rate of pre-shared ebits that are needed to prepare many copies of the given state \cite{Hayden-2001a}.  In the quantum setting, a third party is not introduced into the definition of entanglement cost since by its very nature, quantum entanglement possesses an inherent shielding from external parties.  This latter property is sometime referred to as the \textit{monogamy of entanglement} \cite{Horodecki-2009a}.

The structure of this paper is as follows.  
We begin in Sect. \ref{Sect:LOPC_Reduction} by describing how, for the task at hand, 
all public communication generated during an LOPC protocol can be equivalently 
replaced by pre-shared public correlations at the start of the protocol and no 
further communication.  This provides a significant simplification to the problem 
since a general LOPC protocol can involve multiple rounds of communication.  
In Sect. \ref{Sect:Models}, we introduce in greater detail Wyner's model for 
generating bipartite random variables as well as the tripartite models when 
Charlie is acting either as a collaborator or as an adversary.  The main result 
of this paper is presented in Sect. \ref{Sect:Results} and its proof is given 
in Sects. \ref{Sect:Converse} and \ref{Sect:Coding}.  
The Appendix contains a reformulation of the original protocol given in \cite{Winter-2005a}.

Throughout this paper, random variables will be denoted by capital italic letters $U,V,\cdots$, etc.  The values of the these variables will be written in lower-case $u,v,\cdots$, etc., and a sequence of such values will be denoted as $\mbf{u},\mbf{v},\cdots$, etc.   The distribution of a given random variable $U$ will be interchangeably written by $P_U$ and $P(U)$.  When variables $U$ and $U'$ range over a common alphabet $\mc{U}$, their variational distance (up to a factor of $2$) is given by
\[\|P_U-P_{U'}\|_1:=\sum_{u\in\mc{U}}|P_{U}(u)-P_{U'}(u)|.\]  Finally, when three random variables form a Markov chain, this will be denoted by $W-U-V$, and it indicates that $P(WV|U)=P(W|U)P(V|U)$. Equivalently, its conditional mutual information satisfies $I(W;V|U)=0$.

\section{Replacing Public/Secret Communication by Shared Correlations}
\label{Sect:LOPC_Reduction}
Consider a general LOPC protocol in which Alice and Bob begin with $R$ perfectly correlated bits.  Specifically, Alice (resp. Bob) has variable $V_A$ (resp. $V_B$) such that $V_A=V_B=V$ and $H(V)=R$.  Alice, Bob and Charlie may also have sources of local randomness, but these can be built directly into the local processing of the variables by allowing for stochastic mappings.  The protocol will then involve a sequence of publicly announced messages $M_i$ where $M_i$ is a function of $(V,M_{<i})$.  Here, $M_{<i}=M_1\cdots M_{i-1}$ denotes all previous messages.  At the end of the protocol, the entire communication can be represented by the variable $U$.  Alice and Bob then generate random variables $\hat{X}^n$ and $\hat{Y}^n$, both as the image of some stochastic map applied to $UV$. Thus, the entire protocol can be represented by random variables $\hat{X}^n\hat{Y}^nUV$ whose distribution satisfies
\begin{equation}
  P(\hat{X}^n\hat{Y}^nUV) =P(\hat{X}^n\hat{Y}^n|UV)P(UV)=P(\hat{X}^n|UV)P(\hat{Y}^n|UV)P(UV).
\end{equation}
The particular distribution $P(UV)$ depends on the nature of the LOPC protocol, and here we are using the fact that the computations of $\hat{X}^n$ and $\hat{Y}^n$ are done locally (i.e. independently of each other).  In the collaborative model, Charlie also obtains $\hat{Z}^n$ as a function of $U$, and the distribution is given by
\begin{equation}
  P(\hat{X}^n\hat{Y}^n\hat{Z}^nUV)=P(\hat{X}^n|UV)P(\hat{Y}^n|UV)P(\hat{Z}^n|U)P(UV).
\end{equation}
Hence to simulate the random variables $\hat{X}^n\hat{Y}^n\hat{Z}^nUV$ with no communication, it suffices for all three parties to first share the random variable $U$ (which represents the public correlations of the protocol); additionally, Alice and Bob share the variable $V$ (representing the secret correlations of the protocol).  Conversely, to each distribution $P(VU)$, an LOPC protocol exists with Alice and Bob first sharing secret bits $V_A=V_B=V$ and then broadcasting $U$ according to $P(U|V)$.

\medskip
\begin{remark}
\label{Rem:LOPC}
In the next section, we introduce models in which $U$ and $V$ are uncorrelated.  
While this does not correspond directly to the most general LOPC process, 
when proving the converse in Section \ref{Sect:Converse}, we will allow for 
correlated $U$ and $V$.  Hence, the upper bound we derive on secret and public 
correlation rates will also hold in the LOPC scenario.  
In Section \ref{Sect:Coding}, we show that these lower bounds can be obtained 
by public and private correlations $U$ and $V$ which are, in fact, independent.

Alternatively, one can also directly show, by an operational argument,
that protocols with correlated $U$ and $V$ can always be asymptotically
simulated by one where public and secret correlation are independent.
\end{remark}

\section{Three Models of Correlation Generation}
\label{Sect:Models}
We now describe three different models for generating dependent random variables.  
While both the models in Sects. \ref{Sect:WynerCI} and \ref{Sect:Adversarial} 
have been well-studied, our new contribution is the model described in Sect. \ref{Sect:CoM}.

\subsection{Wyner's Common Information}
\label{Sect:WynerCI}

\begin{figure}[ht]
\begin{center}
\includegraphics[width=0.40\textwidth]{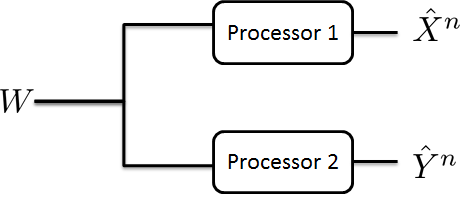}
\caption{Wyner's Common Information scenario.}
\label{Fig:Wyner}
\end{center}
\end{figure}

In this subsection, we review Wyner's notion of common information as 
well as one of its operational interpretations \cite{Wyner-1975a}. 

\medskip
\begin{definition}[Wyner \cite{Wyner-1975a}]
The common information $C(X:Y)$ between two random variables 
$X$ and $Y$ with joint distribution $Q(XY)$ is defined as
\begin{equation}\label{eq_WCI}
  C(X:Y) = \min I(XY;W),
\end{equation}
where the minimization is taken over all triples of random variables $XYW$ so that
\begin{itemize}
  \item the marginal distribution for $X,Y$ is $Q(XY)$;
  \item $X-W-Y$ forms a Markov chain. 
\end{itemize}
Furthermore, the minimum in Eq.~(\ref{eq_WCI}) can be obtained with a 
random variable $W$ ranging over sets of size no greater than $|\cX||\cY|$.
\end{definition}

\medskip
To see why this quantity might capture the notion of ``common information'' 
between $X$ and $Y$, consider the following task.  Alice and Bob have access 
to a common source $W$, and acting independently of one another, they wish 
to process $W$ in different ways so that their final joint distribution is 
a many-copy approximation of the target distribution $Q(XY)$ 
(see Fig.~\ref{Fig:Wyner}).  
The common information is the minimum rate of common randomness $W$ needed to perform this task.

More precisely, we define an $(n,R,\epsilon)$ \emph{source synthesis code} 
to consist of the following:
\begin{itemize}
\item a set $\cW$ with cardinality $\lfloor 2 ^{nR} \rfloor$;
\item conditional probability distributions $P_1^{(n)}(\mbf{x}|w)$ and $P_2^{(n)}(\mbf{y}|w)$, $w\in\cW$, on $\cX^n,\cY^n$, respectively;
\end{itemize}
such that 
\[
  \left\| Q^{(n)}-\hat{P}^{(n)} \right\|_1 \leq \epsilon,
\]
where
\begin{equation}
  \label{Eq:WynerSynthesis}
  \hat{P}^{(n)}(\mbf{x},\mbf{y}) 
       := \frac{1}{|\cW|} \sum_{w\in\cW} P_1^{(n)}(\mbf{x}|w)P_2^{(n)}(\mbf{y}|w).
\end{equation}
We say the rate $R$ is \emph{achievable} if for all $\epsilon>0$ and $n$ 
sufficiently large there exists a source synthesis code $(n,R,\epsilon)$.  
Define the correlation cost of $XY$ as
$C:=\inf \{R: (n,R,\epsilon) \text{ is achievable}\}$.

\medskip
\begin{theorem}[Wyner \cite{Wyner-1975a}]
\label{Thm:Wyner}
For any pair $XY$ of random variables, the minimum achievable rate of a source
synthesis code is given by the common information:
\[
  C=C(X:Y).
\]
\end{theorem}

\medskip
The key ingredient in Wyner's achievability construction is a general result 
saying that for any two random variables $U$ and $W$, the distribution of 
$U^n$ can be reliably simulated by sampling from approximately $2^{nI(U:W)}$ 
sequences among the range of $W^n$ and then applying the channel $P^n_{U|W}$ 
(see Lemma \ref{Lem:Wyner} below).  
Hence if $U=XY$ with $W$ satisfying $X-W-Y$, then this simulation can be done 
locally, as depicted and in Fig. \ref{Fig:Wyner} and described in 
Eq. \eqref{Eq:WynerSynthesis}.  This construction need not be limited to only 
two parties.  For example, one can analogously define the common information 
of three variables $XYZ$ with distribution $Q(XYZ)$ as
\begin{equation}
  \label{Eq:Wyner3CI}
  C(X:Y:Z):=\min I(XYZ;W),
\end{equation}
where the minimization is taken over all variables $XYZW$ so that
\begin{itemize}
  \item the marginal distribution for $XYZ$ is $Q(XYZ)$;
  \item $XYZ$ are conditionally independent variables given $W$.
\end{itemize}
Operationally, and in complete analogy to Wyner's Theorem \ref{Thm:Wyner},
$C(X:Y:Z)$ is the smallest rate of shared random bits $W$ that are needed 
to generate $Q(XYZ)$ when Alice, Bob and Charlie independently process $W$.  
In other words, there are now three channels $P_1^{(n)}(\mbf{x}|w)$, 
$P_2^{(n)}(\mbf{y}|w)$ and $P_3^{(n)}(\mbf{z}|w)$ so that 
\[
  Q^n(\mbf{x},\mbf{y},\mbf{z}) 
      \approx \frac{1}{|\cW|} \sum_{w\in\cW} P_1^{(n)}(\mbf{x}|w)
                                             P_2^{(n)}(\mbf{y}|w)
                                             P_3^{(n)}(\mbf{z}|w)
\]
with $\frac{1}{n}\log|\mc{W}|\leq C(X;Y;Z)+\delta$, for arbitrarily small 
$\delta$.  This is a special case of the more general three-party collaborative 
scenario that we will study below.  Specifically, when the wires connected to 
$V$ are removed in Fig. \ref{Fig:Collaborative}, we recover this scenario of 
Wyner's common information for three parties.

\subsection{Key Cost in Three-Party Adversarial Scenario}
\label{Sect:Adversarial}

\begin{figure}[ht]
\begin{center}
\includegraphics[width=0.60\textwidth]{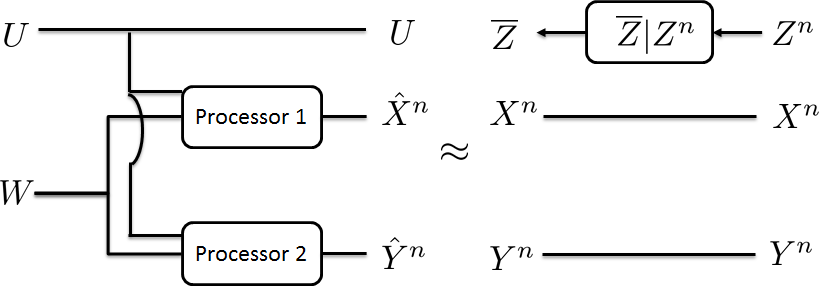}
\caption{Three-Party Adversarial Scenario.}
\label{Fig:Adversarial}
\end{center}
\end{figure}

We now consider three-party distributions generated by two sources of correlations.  First is the adversarial model where Charlie is considered to be a malicious eavesdropper.  Again, let $Q(XYZ)$ be the distribution whose correlation costs we wish to quantify.  In the adversarial model, Alice and Bob start with some initially perfectly correlated bits that are secret from Charlie.  Using LOPC, they wish to generate many copies of $Q(XY)$ so that the total public communication $U$ produced in the protocol gives Charlie no more information about $XY$ than what he has in the distribution $Q(XYZ)$.  In other words, Charlie is able to apply some local processing $\overline{Z}|Z^n$ on her share part of $Q(X^nY^nZ^n)$ so that the resulting distribution is close to the  distribution generated in the LOPC protocol.  There will be some trade-off between the amount of initial secret correlations and the amount of public communication consumed in the protocol.  Intuitively, the more perfectly correlated secret bits that Alice and Bob initially share, the less public communication they will need to generate $Q(XY)$.

By the discussion in Sect. \ref{Sect:LOPC_Reduction}, we can simulate the entire protocol having public communication $U$ by a protocol with no communication but initially shared public correlations.  The resulting scenario is depicted in Fig. \ref{Fig:Adversarial}.  The trade-off between public and private correlations in the task of secrecy formation is formally defined as follows.

\medskip
\begin{definition}[Renner \& Wolf \cite{Renner-2003a}]
\label{def:ach}
For distribution $Q(XYZ)$, an $(n, R_P, R_K,\epsilon)$ \textit{secrecy formation code} is composed of the following:
\begin{itemize}
  \item random variables $(U,V)$ having joint distribution $P(UV)$ over the set 
        $\cW_P\times\cW_K$ with cardinalities $|\cW_P|=\lfloor 2^{nR_P}\rfloor$ 
        and $|\cW_K|=\lfloor 2^{nR_K}\rfloor$ respectively;
  \item conditional distributions on $\cX^n$ and $\cY^n$,
        \[
          P_1^{(n)}(\mbf{x}|\mbf{u},\mbf{v}) \text{ and } P_2^{(n)}(\mbf{y}|\mbf{u},\mbf{v})
          \text{ for } \mbf{u}\in\mc{W}_P,\ \mbf{v}\in\mc{W}_K,
        \]
        which generate random variables $\widehat{X}^n\widehat{Y}^n$ with joint 
        distribution 
        \[
          \widehat{P}(\mbf{x},\mbf{y})
            :=\sum_{\mbf{u}\in\mc{W}_P}\sum_{\mbf{v}\in\mc{W}_K}P_1^{(n)}(\mbf{x}|\mbf{u},\mbf{v})P_2^{(n)}(\mbf{y}|\mbf{u},\mbf{v})P(\mbf{u},\mbf{v});
        \]
  \item a channel $\bar{Z}|Z^n$ such that
        \begin{equation}
          \label{eq:ach}
          \left\| Q(X^nY^n\bar{Z})- \widehat{P}(\widehat{X}^n\widehat{Y}^n U) \right\|_1\leq \epsilon.
        \end{equation}
\end{itemize}
The rate pair $(R_P,R_K)$ is \emph{achievable} if, for all $\epsilon>0$, we can find an $n$ sufficiently large such that there exists a secrecy formation code $(n,R_P,R_K,\epsilon)$.
\end{definition}

\medskip
The public-vs-secret tradeoff function is 
\[
  R_K(R_P)= \inf \bigl\{R_K: (R_P,R_K) \text{ is achievable} \bigr\},
\]
and the secret key cost of the triple $XYZ$ is 
\[
  K_c(X:Y|Z):=\lim_{R_P\to \infty} R_K(R_P).
\]

\medskip
\begin{theorem}[Winter \cite{Winter-2005a}]
\label{thm}
For the secrecy formation of the distribution $Q(XYZ)$, the rate pair $(R_P,R_K)$ is achievable iff there exist random variables $XYZUV$ such that 
\begin{equation}
R_K\geq  I(XY;V|U) \quad \text{and} \quad R_P\geq I(Z;U),
\end{equation}
where the random variables $XYZUV$ satisfy the properties
\begin{enumerate}
\item The $XYZ$ marginal distribution is $Q$;
\item The following Markov chains hold: 
\begin{equation}\label{Markov}
XY-Z-U \quad \text{and} \quad X-UV-Y.
\end{equation}
\end{enumerate}
Furthermore the auxiliary random variables w.l.o.g.~have bounded ranges:
$|\cU|\leq |\cZ|+1$ and $|\cV|\leq |\cX||\cY|$.

In particular,
\[
  K_c(X:Y|Z) = \min \bigl\{ I(XY;V|U): \text{Properties \it{1)} and \it{2)} hold} \bigr\}.
\]
is the secret key cost of $XYZ$ with unlimited public communication.
\qed
\end{theorem}

\subsection{Key Cost in Three-Party Collaborative Scenario}
\label{Sect:CoM}

\begin{figure}[ht]
\begin{center}
\includegraphics[width=0.40\textwidth]{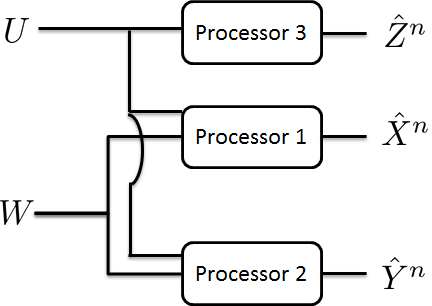}
\caption{Three-Party Collaborative Scenario.}
\label{Fig:Collaborative}
\end{center}
\end{figure}

We now shift perspectives and view Charlie as a collaborator instead 
of an adversary.  As such, for a given distribution $Q(XYZ)$, 
Alice and Bob are no longer content with just generating variables 
$\hat{X}^n\hat{Y}^n$ using LOPC that are close to the target variables $X^nY^n$.  
They also want the public communication $U$ to be sufficiently correlated with $\hat{X}^n\hat{Y}^n$ so that Charlie can locally process $\hat{Z}^n|U$ to jointly produce dependent variables $\hat{X}^n\hat{Y}^n\hat{Z}^n$ that are close to $X^nY^nZ^n$.  Like in the adversarial setting, the public communication can be replaced with initially shared correlations between all the parties (see Fig. \ref{Fig:Collaborative}).  There will also be a trade-off between public and private correlations in the following sense.

\medskip
\begin{definition}
For a given tripartite probability distribution $Q(x,y,z)$ over $\mc{X}\times\mc{Y}\times\mc{Z}$, we define an $(n,R_P,R_K,\epsilon)$ \textit{split-source synthesis code} to be composed of the following:
\begin{itemize}
\item sets $\cW_P$ and $\cW_K$ with cardinalities $\lfloor 2^{nR_P}\rfloor$ and $\lfloor 2^{nR_K}\rfloor$, respectively;
\item conditional probability distributions on $\cX^n$, $\cY^n$ and $\cZ^n$,
\[
  P_1^{(n)}(\bx|\mbf{u},\mbf{v}),\ P_2^{(n)}(\by|\mbf{u},\mbf{v}) \text{ and }  P_3^{(n)}(\mathbf{z}|\mbf{u}) 
                                              \text{ for } \mbf{u}\in \cW_P,\ \mbf{v}\in\cW_K,
\] 
which generate random variables $\widehat{X}^n\widehat{Y}^n\widehat{Z}^n$ with joint distribution 
\begin{equation}
\label{Eq:SynDist1}
\hat{P}^{(n)}(\mbf{x},\mbf{y},\mbf{z})
  = \frac{1}{|\mc{W}_P|}\frac{1}{|\mc{W}_K|}\sum_{\mbf{u}\in\cW_P}\sum_{\mbf{v}\in\cW_K} P_1^{(n)}(\bx|\mbf{u},\mbf{v})P_2^{(n)}(\by|\mbf{u},\mbf{v})P_3^{(n)}(\mathbf{z}|\mbf{u}),
\end{equation}
\end{itemize}
such that  $\| \hat{P}^{(n)}-Q^{(n)} \|_1 \leq \epsilon$.  

The rate pair $(R_P,R_K)$ is \emph{achievable} if, for all $\epsilon>0$, we can find an $n$ sufficiently large such that there exists a split-source synthesis code $(n,R_P,R_K,\epsilon)$.
\end{definition}

\section{Statement of Results}
\label{Sect:Results}
In this section we present our main result: a single-letter characterization of the trade-off between the public and private correlation rate pair in the collaborative scenario of Sect.~\ref{Sect:CoM}.

\medskip
\begin{theorem}
\label{Thm:Main}
For the split-source synthesis of the distribution $Q(XYZ)$, the rate pair $(R_P,R_K)$ is achievable iff there exist random variables $XYZUV$ such that
\begin{align}
\label{Eq:RK}
R_K&\geq I(XY;V|U)\quad \text{and} \quad R_P\geq I(XYZ;U),
\end{align}
where all random variables $XYZUV$ satisfy the following properties
\begin{enumerate}
\item The $XYZ$ marginal distribution is $Q$;
\item $X-VU-Y$ and $XY-U-Z$ form Markov chains.
\end{enumerate}
Furthermore, the random variables $U$ and $V$ in Eq. \eqref{Eq:RK} can be restricted to sets of size no greater than $|\mc{X}||\mc{Y}||\mc{Z}|$ and $|\mc{X}||\mc{Y}|$, respectively.
\end{theorem}

\medskip
Fig.~\ref{Fig:rateregion} illustrates the two-dimensional achievable rate region for the collaborative synthesis of a tripartite distribution. Theorem~\ref{Thm:Main} determines the nontrivial corner point $\alpha$ in Fig.~\ref{Fig:rateregion}: indeed, the public correlation rate at $\alpha$ is given precisely by $C(XY:Z)$, the Wyner common information between $XY$ and $Z$.  Another corner point $\beta$ is when $R_K=0$: here, the problem reduces to the three-party Wyner common information as described in Sect. \ref{Sect:WynerCI}.  Hence, the public correlation rate at $\beta$ is given by $C(X:Y:Z)$.  From this we see that Theorem \ref{Thm:Main} generalizes the notion of Wyner's common information (Theorem \ref{Thm:Wyner}) when $Z$ is trivial.  

\begin{figure}[ht]
\begin{center}
\includegraphics[width=0.5\textwidth]{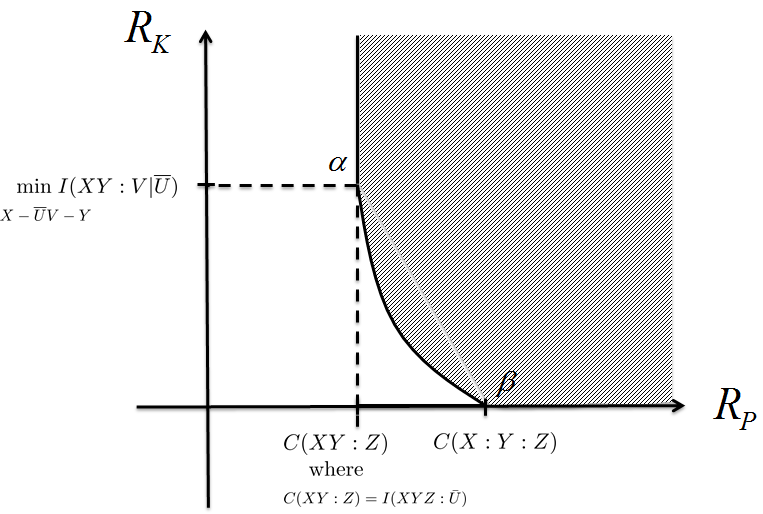}
\caption{Achievable rate region for the collaborative costs of a tripartite distribution using public and secret correlations. The rate pair of the point $\alpha$ (minimal public correlation)
is given after Theorem~\ref{Thm:Main}, as is the point $\beta$ (vanishing private correlation).
Note that, as we can always substitute trivially one bit of private correlation by one
bit of public correlation, the line connecting these points has slope at least as steep
as $-1$.}
\label{Fig:rateregion}
\end{center}
\end{figure}

Any rate pair on the line connecting $\alpha$ with $\beta$ can be achieved by \emph{time-sharing} the two protocols that achieve $\alpha$ and $\beta$, respectively.  What can be said about this line connecting $\alpha$ and $\beta$?  Denote the rate pairs at these points by $(R_P^{(\alpha)},R_K^{(\alpha)})$ and $(R_P^{(\beta)},R_K^{(\beta)})$ respectively.  Since $C(XY:Z)\leq C(X:Y:Z)$, the slope of the connecting line will always be negative.  On the other hand, since the secret correlations of any protocol can always be converted to public correlations, we have that the 
\begin{equation}
  \label{Eq:Rate-Tradeoff}
  R_P^{(\beta)}\leq R_P^{(\alpha)}+ R_K^{(\alpha)}.
\end{equation}
Hence the boundary line connecting $\alpha$ and $\beta$ must have a slope not exceeding $-45^\circ$.  For some distributions, Eq. \ref{Eq:Rate-Tradeoff} is an equality while for others it is not.  In the latter cases, the optimal exchange between private and public correlations is nontrivial and not a simple publication of the private correlations.

\medskip
\noindent\textit{Example 1:}
Let $P_{XYZ}$ be any distribution such that $H(Z|XY)=0$, i.e. $Z$ is a function of $XY$.  Then it is easy to verify that the optimal $R_K$ and $R_P$ tradeoff is one-to-one, and thus the line connecting $\alpha$ and $\beta$ has a slope of $-45^\circ$.  First consider the point $\alpha$.  Here the public correlation rate is given by $R_P^{(\alpha)}=\min_U I(XYZ;U)$ such that $XY-U-Z$.  It is easy to show that $I(XYZ;U)\geq I(XY;Z)$, and this lower bound can be attained by $U=Z$ since $H(Z|XY)=0$.  Therefore, $R_P^{(\alpha)}=H(Z)$ and $R_K^{(\alpha)}=\min_V I(XY;V|Z)$ where $X-ZV-Y$.  Let  $\hat{V}$ denote the variable attaining this minimum.  Now consider the point $\beta$ where $R_K^{(\beta)}=0$ and $R_P^{(\beta)}=\min_U I(XYZ;U)$ such that $XYZ$ are conditionally independent given $U$.  Then
\begin{align}
\min_U I(XYZ;U)&=\min_U \left(I(Z;U)+I(XY;U|Z)\right)\notag \\
&\geq I(Z;XY)+\min_U I(XY;U|Z)\notag\\
&\geq I(Z;XY)+  I(XY;\hat{V}|Z)\notag\\
&=H(Z)+  I(XY;\hat{V}|Z),
\end{align}
where the first inequality is data processing since $Z-U-XY$; the second inequality is obtained since if $XYZ$ are conditionally independent given $U$, then $X-ZU-Y$; and the last equality follows because $H(Z|XY)=0$.  Combining with Eq. \eqref{Eq:Rate-Tradeoff}, we see that
\[R_P^{(\beta)}= R_P^{(\alpha)}+ R_K^{(\alpha)}.\]

\medskip
\noindent\textit{Example 2:}
Next, we consider a very simple distribution $P_{XYZ}$ over $\{0,1,2\}^{\times 3}$ 
with the only nonzero values being $P(x,y,z)=\frac{1}{5}$ for $P(2,2,z)$ with 
$z\in\{0,1,2\}$ and $P(x,y,2)$ with $x=y\in\{0,1\}$.  
This belongs to a more general class of ``L-shaped'' distributions studied 
by Witsenhausen \cite[Thm.~7]{Witsenhausen-1976a}. From his result, the 
common information $C(XY:Z)$ is found to be
\begin{align}
R_P^{(\alpha)}=C(XY:Z)=\tfrac{4}{5}\left(\log\tfrac{4}{5}-\log\tfrac{2}{5}\right)\approx .693.
\end{align}
This is computed from the optimal decomposition of $P_{XYZ}$ into conditionally independent parts:
\[\tfrac{1}{5}\left(\begin{smallmatrix}0&0&1\\0&0&1\\1&1&1\end{smallmatrix}\right)=\tfrac{1}{10}\left(\begin{smallmatrix}2\\2\\1\end{smallmatrix}\right)(0\;0\;1)+\tfrac{1}{10}\left(\begin{smallmatrix}0\\0\\1\end{smallmatrix}\right)(2\;2\;1).\]
Here we have grouped $XY$ into one variable ranging over $\{0,1,2\}$ so that the $(i,j)$ element of the matrix is $P_{XYZ}(i,i,j)$.  From this decomposition, we see that
\begin{equation}
R_K^{(\alpha)}=\min_V I(XY;V|U)=-\tfrac{1}{2}\left(\tfrac{4}{5}\log\tfrac{2}{5}+\tfrac{1}{5}\log\tfrac{1}{5}\right).
\end{equation}
For the corner point $\beta$, we observe that $XYZ$ are conditionally independent given $X$.  Hence,
\begin{equation}
R_P^{(\beta)}\leq H(X)=-\tfrac{3}{5}\log\tfrac{3}{5}-\tfrac{2}{5}\log\tfrac{1}{5}.
\end{equation}
Hence,
\begin{equation}
R_P^{(\alpha)}+R_K^{(\alpha)}\approx 1.08 >R_P^{(\beta)}\approx .950,
\end{equation}
and so the optimal private to public exchange for this distribution is not achieved by simply publicly revealing private correlations.

\bigskip
In the next two sections, we will prove Theorem \ref{Thm:Main}, first the
converse (Sect. \ref{Sect:Converse}), then the direct part 
(Sect. \ref{Sect:Coding}).

\section{Converse}
\label{Sect:Converse}
Here we derive lower bounds that hold for more general models than a synthesis code.  Specifically, we assume that $\hat{X}\hat{Y}\hat{Z}UV$ is given along with conditional probabilities $P_1^{(n)}$, $P_2^{(n)}$, and $P^{(n)}_3$ such that the generated distribution
\begin{equation}
\label{Eq:SynDist2}
\hat{P}^{(n)}(\mbf{x},\mbf{y},\mbf{z})=\sum_{u\in\cW_P}\sum_{v\in\cW_K} P(u,v) P_1^{(n)}(\bx|u,v)P_2^{(n)}(\by|u,v)P_3^{(n)}(\mathbf{z}|u), 
\end{equation}
satisfies $\|\hat{P}^{(n)}-Q^{(n)}\|_1 \leq \epsilon$. Note that the local processing in the secrecy formation protocol imposes that $\hat{X}^n-VU-\hat{Y}^n$ and $\hat{X}^n\hat{Y}^n-U-\hat{Z}^n$ form Markov chains. However, unlike a synthesis code defined in Sect.~\ref{Sect:CoM}, we do not require that $U$ and $V$ are independent.  This relaxation enables to simulate LOPC protocols as discussed in Remark \ref{Rem:LOPC}.


Following the argument in \cite{Winter-2005a}, monotonicity and the chain rule allow us to write
\begin{align}
R_K\geq \frac{1}{n}\log|\mc{W}_K|&\geq\frac{1}{n}I(\hat{X}^n\hat{Y}^n;V|U)\notag\\
&=\sum_{j=1}^n\frac{1}{n}I(\hat{X}_j\hat{Y}_j;V|U\hat{X}_{<j}\hat{Y}_{<j})\notag\\
&=I(\hat{X}_J\hat{Y}_J;V|UJ\hat{X}_{<J}\hat{Y}_{<J})\notag\\
&=I(\hat{X}_J\hat{Y}_J;V|UJ\hat{X}_{<J}\hat{Y}_{<J}\hat{Z}_{<J}),
\end{align}
where $J\in\{1,\cdots n\}$ is a uniformly distributed variable and the last equality follows from the conditional independence $\hat{X}^n\hat{Y}^n-U-\hat{Z}^n$ (see Proposition~\ref{Prop:MarkovReduced} below).  We next introduce the following random variables $\hat{W}:=J\hat{X}_{<J}\hat{Y}_{<J}\hat{Z}_{<J}$ and $W:=JX_{<J}Y_{<J}Z_{<J}$ and the variables $\tilde{U}\in\mc{W}_P$, $\tilde{V}\in\mc{W}_K$, $\tilde{X}\in\mc{X}$, $\tilde{Y}\in\mc{Y}$, and $\tilde{Z}\in\mc{Z}$ defined through the joint distributions
\begin{align}
\label{Eq:NewVariables}
P(\tilde{X}\tilde{Y}\tilde{Z}\tilde{V}\tilde{U}|W)&=\hat{P}(\hat{X}_J\hat{Y}_J\hat{Z}_JVU|\hat{W}).
\end{align}
Then
\begin{align}
\label{Eq:ProbApprox1}
\hat{P}(\hat{X}_J\hat{Y}_J\hat{Z}_J\hat{W})&=\hat{P}(\hat{X}_J\hat{Y}_J\hat{Z}_J|\hat{W})\hat{P}(\hat{W})=P(\tilde{X}\tilde{Y}\tilde{Z}|W)\hat{P}(\hat{W})\notag\\
&=P(\tilde{X}\tilde{Y}\tilde{Z}W)+P(\tilde{X}_J\tilde{Y}\tilde{Z}_J|W)[\hat{P}(\hat{W})-P(W)].
\end{align}
At the same time, applying the triangle inequality to 
$\|\hat{P}^{(n)}(\hat{X}^n\hat{Y}^n\hat{Z}^n)-Q^{(n)}(X^nY^nZ^n)\|_1 \leq \epsilon$ 
allows us to conclude that
\begin{equation}
\label{Eq:ProbApprox2}
\|\hat{P}(\hat{X}_J\hat{Y}_J\hat{Z}_J\hat{W})-P(X_JY_JZ_JW)\|_1\leq \epsilon,
\end{equation}
and therefore $\|P(\hat{W})-P(W)\|_1 \leq \epsilon$.  Combining the latter with Eqns. \eqref{Eq:ProbApprox1} and \eqref{Eq:ProbApprox2} yields
\begin{align}
\label{Eq:ProbApprox3}
\|P(\tilde{X}\tilde{Y}\tilde{Z}W)-P(X_JY_JZ_JW)\|_1\leq 2\epsilon.
\end{align}
Since $X^nY^nZ^n$ are i.i.d., the marginal distribution of $P(X_JY_JZ_JW)$ is $Q(XYZ)$.  Hence, the previous inequality gives
\begin{equation}
\label{Eq:ProbApprox4}
\|P(\tilde{X}\tilde{Y}\tilde{Z})-Q(XYZ)\|_1\leq 2\epsilon.
\end{equation}
Eq. \eqref{Eq:NewVariables} also gives that $P(\tilde{V}|\tilde{U}W)=\hat{P}(V|U\hat{W})$ and $P(\tilde{X}\tilde{Y}\tilde{Z}|\tilde{U}W)=\hat{P}(\hat{X}_J\hat{Y}_J\hat{Z}_J|U\hat{W})$.  The Markov conditions $\hat{X}_J\hat{Y}_J-U\hat{W}-\hat{Z}_J$, $\hat{X}_J-VU\hat{W}-\hat{Y}_J$, and $\hat{Z}_J-U\hat{W}-V$ therefore 
imply 
\begin{equation}
\label{Eq:MarkovNew}
  \tilde{X}\tilde{Y}-\tilde{\tilde{U}}-\tilde{Z}
  \quad\text{and}\quad 
  \tilde{X}-\tilde{V}\tilde{\tilde{U}}-\tilde{Y},
\end{equation}
where $\tilde{\tilde{U}}:=\tilde{U}W$.  


{From Eq. \eqref{Eq:NewVariables}, we have $\hat{P}(\hat{X}_J\hat{Y}_JV|U\hat{W})P(U|\hat{W})=P(\tilde{X}\tilde{Y}\tilde{V}|\tilde{U}W)P(\tilde{U}|W)$ which further implies that $\hat{P}(\hat{X}_J\hat{Y}_JV|U\hat{W})=P(\tilde{X}\tilde{Y}\tilde{V}|\tilde{U}W)$ since $\hat{P}(U|\hat{W})=P(\tilde{U}|W)$, again by Eq.  \eqref{Eq:NewVariables}.  Thus, for each fixed value of $w$, we have that 
\[I(\hat{X}_J\hat{Y}_J;V|U\hat{W}=w)\hat{P}(U|\hat{W}=w)=I(\tilde{X}\tilde{Y};\tilde{V}|\tilde{U}W=w)P(\tilde{U}|W=w).\]
Multiply both sides by $P(\hat{W}=w)$ and take the sum.  Using the fact that    $\|P(\hat{W})-P(W)\|_1\leq \epsilon$ and the triangle inequality lead to:}

\begin{align}
|I(\hat{X}_J\hat{Y}_J;V|U\hat{W})-I(\tilde{X}\tilde{Y};\tilde{V}|\tilde{\tilde{U}})|\leq\epsilon\log|\mc{X}||\mc{Y}|,\notag
\end{align}
hence,
\begin{equation}
\label{Eq:RKBound}
R_K\geq I(\tilde{X}\tilde{Y};\tilde{V}|\tilde{\tilde{U}})-\epsilon\log|\mc{X}||\mc{Y}|.
\end{equation}
By the same arguments, we can bound the public communication as  
\begin{align}
R_P\geq\frac{1}{n}\log|\mc{W}_P|&\geq \frac{1}{n}I(\hat{X}^n\hat{Y}^n\hat{Z}^n;U)= I(\hat{X}_J\hat{Y}_J\hat{Z}_J:U|\hat{W})\notag\\
&\geq I(\tilde{X}\tilde{Y}\tilde{Z};\tilde{U}|W)-\epsilon\log|\mc{X}||\mc{Y}||\mc{Z}|\notag\\
&=I(\tilde{X}\tilde{Y}\tilde{Z};\tilde{\tilde{U}})-I(\tilde{X}\tilde{Y}\tilde{Z};W)-\epsilon\log|\mc{X}||\mc{Y}||\mc{Z}|.\notag
\end{align}
To bound the term $I(\tilde{X}\tilde{Y}\tilde{Z};W)=H(\tilde{X}\tilde{Y}\tilde{Z})-\sum_{w}H(\tilde{X}\tilde{Y}\tilde{Z}|W=w)P(W=w)$, we recall a well-known continuity relation: Any two random variables $A$ and $A'$ ranging over $\mc{A}$ with $\delta:=\|P(A)-P(A')\|_1\leq 1/2$ satisfy $|H(A)-H(A')|\leq-\delta\log\tfrac{\delta}{|\mc{A}|}$ \cite{Csiszar-2011a}.  Therefore, using Eq. \eqref{Eq:ProbApprox3} and the fact that $I(X_JY_JZ_J;W)=0$, we readily obtain
\begin{equation}
\label{Eq:RPBound}
R_P\geq I(\tilde{X}\tilde{Y}\tilde{Z};\tilde{\tilde{U}})+4\epsilon\log 2\epsilon-5\epsilon\log|\mc{X}||\mc{Y}||\mc{Z}|.
\end{equation}

At this point we have constructed random variables $\tilde{X}\tilde{Y}\tilde{Z}\tilde{V}\tilde{\tilde{U}}$ that satisfy Eqns. \eqref{Eq:ProbApprox4}--\eqref{Eq:RPBound}.  By Lemma \ref{Lem:Carathedory}, we can assume without loss of generality that $\tilde{V}$ and $\tilde{\tilde{U}}$ range over sets of size no greater than $|\mc{X}||\mc{Y}||\mc{Z}|$.  Hence, the set of random variables satisfying Eqns. \eqref{Eq:ProbApprox4}--\eqref{Eq:RPBound} is compact, and therefore a limit point will exist which also satisfies these constraints when taking $\epsilon\to 0$.  
This proves the lower bound of Theorem \ref{Thm:Main}. \qed

\medskip
\begin{proposition}
\label{Prop:MarkovReduced}
If $n$-part random variables $A^n$ and $B^n$ satisfy $A^n-C-B^n$, then the reduced variables 
$A_j$ and $B_{k}$ satisfy $A_j-CA_{<j}B_{<k}-B_{k}$ for any 
$1\leq j,k\leq n$, where $A_{<j}=A_1\ldots A_{j-1}$ and likewise 
$B_{<k} = B_1\ldots B_{k-1}$.
\end{proposition}
\begin{proof}
Consider the marginal distribution $A_jA_{<j}-C-B_{k}B_{<k}$.  Then 
\begin{align}
P(A_jB_{k}|CA_{<j}B_{<k})&=\frac{P(A_jA_{<j}B_{k}B_{<k}|C)}{P(A_{<j}B_{<k}|C)},\\
&=\frac{P(A_jA_{<j}|C)}{P(A_{<j}|C)}\frac{P(B_{k}B_{<k}|C)}{P(B_{<k}|C)},\notag\\
&=P(A_j|A_{<j}C)P(B_{k}|B_{<k}C).\notag
\end{align}
Therefore, $A_j-CA_{<j}B_{<k}-B_{k}$.
\end{proof}

\medskip
\begin{lemma}
\label{Lem:Carathedory}
Suppose that $XYZUV$ are random variables with $UV$ ranging over $\mc{U}\times\mc{V}$ such that $XY-U-Z$ and $X-UV-Y$.  Then there exists random variables $X'Y'Z'V'U'$ satisfying the same Markov chain and 
\begin{subequations}
\label{Eq:RVcons}
\begin{align}
I(X'Y'Z';U')&=I(XYZ;U),\\
P(X'Y'Z')&=P(XYZ),\\
I(X'Y';V'|U')&\leq I(XY;V|U),
\end{align}
\end{subequations}
with $U'$ and $V'$ ranging over sets $\mc{U}'$ and $\mc{V'}$ of sizes $|\mc{U}'|\leq |\mc{X}||\mc{Y}||\mc{Z}|+1$ and $|\mc{V}'|\leq |\mc{X}||\mc{Y}||\mc{Z}|$.  Furthermore, if $Z-U-V$ also holds, then the size of $\mc{V}'$ can be further reduced to $|\mc{V}'|\leq |\mc{X}||\mc{Y}|$.
\end{lemma}
\begin{proof}
For the given distribution $P(XYZUV)$, let $\{P(XYZ|vu)\}_{v\in\mc{V},u\in\mc{U}}$ and $\{P(XYZ|u)\}_{u\in\mc{U}}$ be the associated conditional distributions.  For each fixed $u\in\mc{U}$, let $\Lambda_u$ be the collection of conditional distributions over $\mc{V}$ such that $\lambda(v|u)\in\Lambda_u$ if $\sum_{v}P(XYZ|uv)\lambda(v|u)=P(XYZ|u)$.  This represents a total of $N=|\mc{X}||\mc{Y}||\mc{Z}|-1$ linear constraints on the $\lambda(v|u)$ (note if $Z-U-V$ also holds, then $\sum_{v}P(XYZ|uv)\lambda(v|u)=P(XYZ|u)$ reduces to $\sum_{v}P(XY|uv)\lambda(v|u)=P(XY|u)$ which represents a total of $N=|\mc{X}||\mc{Y}|-1$ linear constraints on the $\lambda(v|u)$).  Now $\Lambda_u$ is convex and the set $\{\sum_{v}H(XY|u,v)\lambda(v|u):\lambda\in\Lambda_u\}$ will obtain both its maximum and minimum at an extreme point of $\Lambda_u$.  Then an application of Carath\'{e}odory's Theorem (Lemma \ref{Lem:Carathedory2}) guarantees that such an extreme point is a distribution over $\mc{V}$ with no more than $N+1$ nonzero probability values \cite{Klee-1963a}.  Hence by a conditional relabeling of the $v$, we have a subset $\mc{V}'\subset\mc{V}$ with $|\mc{V}'|\leq N+1$ and a collection of conditional distributions $\lambda'(v|u)$ over $\mc{V}'$ such that
\begin{subequations}
\begin{align}
\label{Eq:Dubin1a}
\sum_{v\in\mc{V}'}P(XYZ|uv)\lambda'(v|u)&=\sum_{v\in\mc{V}}P(XYZ|u)\quad\forall u\in\mc{U},\\
\label{Eq:Dubin1b}
\sum_{v\in\mc{V}'}H(XY|uv)\lambda'(v|u)&\geq\sum_{v\in\mc{V}}H(XY|V,u)\quad\forall u\in\mc{U}.
\end{align}
\end{subequations}

We now perform a similar argument by letting $\Gamma$ be the set of all distributions over $\mc{U}$ such that $\gamma(u)\in\Gamma$ if \\$\sum_{u\in\mc{U}}P(XYZ|u)\gamma(u)=P(XYZ)$ and $\sum_{u\in\mc{U}}H(XYZ|U=u)\gamma(u)=H(XYZ|U)$.  This represents $|\mc{X}||\mc{Y}||\mc{Z}|$ linear constraints on $\gamma$, and we seek the minimum value of the set $\{\sum_{u}[H(XY|u)-H(XY|V,u)]\gamma(u):\gamma\in\Gamma\}$.  A second application of Carath\'{e}odory's Theorem ensures the existence of a distribution $\gamma'(u)$ ranging over $\mc{U}'\subset\mc{U}$ with $|\mc{U}'|\leq |\mc{X}||\mc{Y}||\mc{Z}|+1$ for which
\begin{subequations}
\begin{align}
\label{Eq:Dubin2a}
\sum_{u\in\mc{U}'}P(XYZ|u)\gamma'(u)&=P(XYZ),\\
\label{Eq:Dubin2b}
\sum_{u\in\mc{U}'}H(XYZ|u)\gamma'(u)&=H(XYZ|U),\\
\label{Eq:Dubin2c}
\sum_{u\in\mc{U}}[H(XY|u)-H(XY|V,u)]\gamma'(u)&\leq H(XY|U)-H(XY|UV)=I(XY;V|U).
\end{align}
\end{subequations}
This completes the construction of random variables $X'Y'Z'U'V'$ whose joint distribution is given by $P(X'Y'Z'U'V'):=P(XYZ|uv)\lambda'(v|u)\gamma'(u)$.  By its definition and by Eq. \eqref{Eq:Dubin1a}, $X'Y'Z'U'V'$ inherits whatever Markov chain properties are present in $XYZUV$.  Eq. \eqref{Eq:Dubin2a} gives $P(X'Y'Z')=P(XYZ)$, and combining this with Eq. \eqref{Eq:Dubin2b} yields $I(X'Y'Z';U')=I(XYZ;U)$.  Finally, combining Eq. \eqref{Eq:Dubin2c} and Eq. \eqref{Eq:Dubin1b} gives
\begin{equation}
I(X'Y';V'|U')\leq I(XY;V|U),
\end{equation}
concluding the proof.
\end{proof}

\medskip
\begin{lemma}[Carath\'{e}odory's Theorem~\cite{Rockafellar-1996a}]
\label{Lem:Carathedory2}
Let $S$ be a subset of $\mathbb{R}^n$ and $\text{conv}(S)$ its convex hull.  
Then any $x \in \text{conv}(S)$ can be expressed as a convex combination of at 
most $n+1$ elements of $S$.
\qed
\end{lemma}

\medskip
\begin{remark}
An application of Carath\'{e}odory's Theorem shows that if the elements of $\text{conv}(S)$ are further required to satisfy $d$ linear constraints, then the resulting set is convex with extreme points being convex combinations of at most $d+1$ extreme points of $\text{conv}(S)$ \cite{Klee-1963a}.
\end{remark}

\section{Achievability}
\label{Sect:Coding}
Let $XYZUV$ be random variables with joint distribution $P(XYZUV)$ satisfying (1) $Q(XYZ)=P(XYZ)$ and (2) $X-VU-Y$ and $XY-U-Z$.  In what follows, we 
let $T_{[U]_\delta}^n$ denote the set of all $\delta$-typical sequences 
with respect to random variable $U$ having distribution $P(U)$. 
Recall that a sequence $\mbf{u}\in\mc{U}^n$ is $\delta$-typical if 
$\left|\frac{N(u|\mbf{u})}{n}-P(u)\right| \leq \delta$ for all 
$u\in\mc{U}$ \cite{Csiszar-1978a}.

Our code makes repeated use of Wyner's original code.  
The following is proven in \cite{Wyner-1975a}, where here we have modified 
the statement using Pinsker's inequality, $D(P_1||P_2)\geq \tfrac{1}{2}\|P_1-P_2\|_1^2$, 
to obtain a bound on the variational distance. See also later works by Han
and Verd\'u \cite{HanVerdu-1992,HanVerdu-1993} and Ahlswede \cite{Ahlswede-2006},
where more general versions were proved (Ref. \cite{AhlswedeWinter-2002} 
contains a quantum analogue).

\begin{lemma}[Wyner {\cite[Thm. 6.3]{Wyner-1975a}}]
\label{Lem:Wyner}
Let $AB$ be random variables over $\mc{A}\times\mc{B}$ with joint distribution $P(AB)$,
and let $R> I(A;B)$.  
For $\epsilon>0$ and sufficiently large $n$, there exists a subset 
$\beta\subset T^n_{[B]_\delta}\subset\mc{B}^n$ of size 
$|\beta|=\lfloor 2^{nR}\rfloor$ such that for
\begin{equation}
  \hat{P}^{(n)}(\mbf{a})
    =\frac{1}{|\beta|}\sum_{\mbf{b}\in\beta}P^{(n)}(\mbf{a}|\mbf{b})\quad\text{for}\quad\mbf{a}\in\mc{A}^n,
\end{equation}
it holds that $\|P^{(n)}(A^n)-\hat{P}^{(n)}(A^n)\|_1 \leq \epsilon$.
\end{lemma}

\medskip
Identify $A:=XYZ$ and $B:=U$ in Lemma \ref{Lem:Wyner}. Thus, for $n$ sufficiently large, we can find a subset $\mc{W}_P\subset\mc{U}^n$ such that $|\mc{W}_P|=\lfloor 2^{nR_P}\rfloor$ with
\begin{equation}
\label{Eq:PublictCodeRate}
R_P=I(XYZ:U)+\delta,
\end{equation} 
and 
\begin{equation}
\label{Eq:ErrorApprox1}
\left\|Q^{(n)}(\mbf{x},\mbf{y},\mbf{z})-\frac{1}{|\mc{W}_P|}\sum_{\mbf{u}\in\mc{W}_P}P^{(n)}(\mbf{x},\mbf{y}|\mbf{u})P^{(n)}(\mbf{z}|\mbf{u})\right\|_1 \leq \epsilon,
\end{equation}
where we have used the Markov chain $XY-U-Z$. Here, $P^{(n)}(\mbf{z}|\mbf{u})$ will be the encoder employed by the collaborative third party.

We next consider the term
\begin{equation}
P^{(n)}(\mbf{x},\mbf{y}|\mbf{u})=\prod_{u\in\mc{U}}P^{(N(u|\mbf{u}))}(\mbf{x}_u,\mbf{y}_u|u),
\end{equation}
where $(\mbf{x}_u,\mbf{y}_u)$ is a sequence of length $N(u|\mbf{u})$ that occurs with the event $U=u$.  Knowing that $\mbf{u}\in T^n_{[U]_\delta}$, we let $n_u:=\lfloor n(P(u)+\delta)\rfloor$, and for each $u\in\mc{U}$, we apply Lemma \ref{Lem:Wyner} on the conditional distribution $P(XYV|U=u)$ with the choice $A:=XY$ and $B:=V$.  This will generate a collection of codeword sets $\alpha_u$, each with respective size $|\alpha_u|=\lfloor 2^{n_u(I(XY:V|U=u)+\delta)}\rfloor$.  Furthermore,
\begin{align}
\hat{P}^{(n_u)}(\mbf{x}_u,\mbf{y}_u|u):=\frac{1}{|\alpha_u|}\sum_{\mbf{v}\in\alpha_u}P^{(n_u)}(\mbf{x}_u,\mbf{y}_u|u,\mbf{v}_u)=\frac{1}{|\alpha_u|}\sum_{\mbf{v}\in\alpha_u}P^{(n_u)}(\mbf{x}_u|u,\mbf{v}_u)P^{(n_u)}(\mbf{y}_u|u,\mbf{v}_u)
\end{align}
satisfies $||\hat{P}^{(n_u)}(\mbf{x}_u,\mbf{y}_u|u)-P^{(n_u)}(\mbf{x}_u,\mbf{y}_u|u)||_1<\epsilon$.  In the previous equation, the Markov chain $X-UV-Y$ has been employed.  For each $u\in\mc{U}$ and $\mbf{v}_u\in\alpha_u$, let $\hat{P}^{(N(u|\mbf{u}))}(\mbf{x}_u|u,\mbf{v}_u)$ denote the marginal distribution obtained from $P^{(n_u)}(\mbf{x}_u|u,\mbf{v}_u)$ by summing over the last $n_u-N(u|\mbf{v}_u)$ events.  Let $\hat{P}^{(N(u|\mbf{u}))}(\mbf{y}_u|u,\mbf{v}_u)$ be defined likewise.  Thus,  
\[||\hat{\hat{P}}^{N(u|\mbf{u})}(\mbf{x}_u,\mbf{y}_u|u)-P^{N(u|\mbf{u})}(\mbf{x}_u,\mbf{y}_u|u)||_1<\epsilon,\] where
\begin{equation}
\hat{\hat{P}}^{N(u|\mbf{u})}(\mbf{x}_u,\mbf{y}_u|u):=\frac{1}{|\alpha_u|}\sum_{\mbf{v}_u\in\alpha_u}\hat{P}^{(N(u|\mbf{u}))}(\mbf{x}_u|u,\mbf{v}_u)\hat{P}^{(N(u|\mbf{u}))}(\mbf{y}_u|u,\mbf{v}_u).
\end{equation}
We now paste together the different codes to form the code set $\mc{W}_K=\prod_{u\in\mc{U}}\alpha_u$.  For any typical $\mbf{u}$ and $\mbf{v}\in\mc{W}_K$, we define the local generators
\begin{align}
\hat{\hat{P}}^{(n)}(\mbf{x}|\mbf{u},\mbf{v})&:=\prod_{u\in\mc{U}}\hat{P}^{(N(u|\mbf{u}))}(\mbf{x}_u|u,\mbf{v}_u)\notag\\
\hat{\hat{P}}^{(n)}(\mbf{y}|\mbf{u},\mbf{v})&:=\prod_{u\in\mc{U}}\hat{P}^{(N(u|\mbf{u}))}(\mbf{y}_u|u,\mbf{v}_u),
\end{align}
which satisfy
\begin{equation}
 \left\|P^{(n)}(\mbf{x},\mbf{y}|\mbf{u})-\frac{1}{|\mc{W}_K|}\sum_{\mbf{v}\in\mc{W}_K}\hat{\hat{P}}^{(n)}(\mbf{x}|\mbf{u},\mbf{v})\hat{\hat{P}}^{(n)}(\mbf{y}|\mbf{u},\mbf{v})\right\|_1<|\mc{U}|\epsilon.
\end{equation}
The size of $\mc{W}_K$ is bounded by
\begin{align}
\log |\mc{W}_K|=\sum_{u\in\mc{U}} \log |\alpha_u|&\leq \sum_{u\in\mc{U}}n_u(I(XY:V|U=u)+\delta)\notag\\
&\leq n(I(XY:V|U)+O(\delta)).
\end{align}
Combining this simulation of $P^{(n)}(\mbf{x},\mbf{y}|\mbf{u})$ with Eq. \eqref{Eq:ErrorApprox1} gives the final error bound
\begin{equation}
  \label{Eq:ErrorApprox3}
  \left\| Q^{(n)}(\mbf{x},\mbf{y},\mbf{z})-\hat{P}^{(n)}(\mbf{x},\mbf{y},\mbf{z}) \right\|_1
     \leq (|\mc{U}|+1)\epsilon.
\end{equation}
where
\[
  \hat{P}^{(n)}(\mbf{x},\mbf{y},\mbf{z})
      =\frac{1}{|\mc{W}_P|}\frac{1}{|\mc{W}_K|}
          \sum_{\mbf{u}\in\mc{W}_P}
          \sum_{\mbf{v}\in\mc{W}_K} \hat{\hat{P}}^{(n)}(\mbf{x}|\mbf{u},\mbf{v})
                                    \hat{\hat{P}}^{(n)}(\mbf{y}|\mbf{u},\mbf{v})P^{(n)}(\mbf{z}|\mbf{u}).
\]
Since $|\mc{U}|\leq |\mc{X}||\mc{Y}||\mc{Z}|$ and $|\mc{V}|\leq |\mc{X}||\mc{Y}|$, 
the bounds on $R_P$, $R_K$ and $\|Q^{(n)}-\hat{P}^{(n)}\|_1$ can be made arbitrarily 
close to $I(XYZ:U)$, $I(XY:V|U)$ and zero, respectively.

\qed

\section{Conclusion}
In this paper we have introduced the problem of tripartite correlation 
generation using public and private correlations.  This can be seen as 
a collaborative alternative to the cryptographic problem of secrecy formation.  
We have found that despite the two different natures of the problem, 
the optimal secret correlation rates have a very similar structure.  
We have completely characterized the public-vs-private rate region for 
the collaborative scenario.  One point of interest is when the public 
correlation rate is minimum (point $\alpha$ in Fig. \ref{Fig:rateregion}), 
and another is when the secret correlation rate is zero (point $\beta$ in 
Fig. \ref{Fig:rateregion}).  We have shown that the optimal exchange of 
private to public correlations does not always involve a trivial 
publicizing of private information.  However, it is an interesting open 
problem to determine the slope of the line connecting $\alpha$ and $\beta$ 
for a general distribution, in particular to understand what limits
there are, if any, on the exchange rate of private to public correlation rate.

Further afield, following the example of Ref. \cite{Winter-2005a}, one
could ask for the benefit of using entanglement instead of, or
in addition to, the private and public shared randomness. We leave
this and other questions for future investigations.

\section*{Acknowledgments}
EC was supported by the National Science Foundation (NSF) Early CAREER Award No.~1352326.  
MH is supported by an ARC Future Fellowship under Grant FT140100574. 
AW was supported by the European Commission (STREP ``RAQUEL''), 
the European Research Council (Advanced Grant ``IRQUAT''), 
and the Spanish MINECO (project FIS2008-01236) with FEDER funds.

\appendix

\section{Achievability Proof of Secret Key Cost in the Adversarial Scenario}
Here we review the achievability component of Theorem~\ref{thm}.  The coding for Alice and Bob is the same as described in Section \ref{Sect:Coding}.  Let $XYZUV$ be random variables obtaining the minimum in Theorem~\ref{thm}, 
and let $P(XYZUV)$ denote their joint distribution so that the marginal on 
$XYZ$ is $Q(XYZ)$. The public correlation (communication) is $U^n$ 
($n$ i.i.d. realisations of $U$). Let us restrict attention to the typical 
subset for which the relative frequency of each letter $u$ in $U^n$ is close 
to $P(U=u)$; in particular, 
$\left|P(U=u)-\frac{N(u|\mbf{u})}{n}\right|\leq \delta$. For the set of 
positions where $u$ occurs, we can employ Lemma~\ref{Lem:Wyner}. 

Consider $A=XY|_{U=u}$ and $B=V|_{U=u}$  in Lemma \ref{Lem:Wyner}.
Thus for $n_u:=\lfloor n(P(u)+\delta)\rfloor$ sufficiently large, we can 
find a subset $\alpha_{u}\subset\mc{V}^{n_{u}}$ such that 
$|\alpha_{u}|=\lfloor 2^{n_u(I(XY:V|U=u)+\delta)}\rfloor$
with
\begin{equation}
\label{Eq:WErrorApprox1}
\left\|P^{(n_{u})}(X^{n_u}Y^{n_u}|u)-\hat{P}^{(n_{u})}(X^{n_u}Y^{n_u}|u)\right\|_1 \leq \epsilon,
\end{equation}
where 
\[
\hat{P}^{(n_{u})}(X^{n_u}Y^{n_u}|u):=\frac{1}{|\alpha_{u}|}\sum_{\mbf{v}\in\alpha_{u}}P^{(n_{u})}(X^{n_u}Y^{n_u}|u,\mbf{v}).
\]
We paste these codes together and define local channels for Alice and Bob $\hat{\hat{P}}^{(n)}(X^n|\mbf{u},\mbf{v})\hat{\hat{P}}^{(n)}(Y^n|\mbf{u},\mbf{v})$ which, for each $\mbf{u}\in T^n_{[U]_\delta}$, samples from the concatenated code and discards the extra letter occurrences not found in $\mbf{u}$ (see Section  \ref{Sect:Coding}).  With $\mc{W}_K$ denoting the set of code words, this generates the simulation 
\[\tilde{P}^{(n)}(X^nY^n|\mbf{u}):=\frac{1}{|\mc{W}_K|}\sum_{\mbf{v}\in\mc{W}_K}\hat{\hat{P}}^{(n)}(X^n|\mbf{u},\mbf{v})\hat{\hat{P}}^{(n)}(Y^n|\mbf{u},\mbf{v})\]
which satisfies 
\begin{equation}
\label{Eq:EndStep1}
 \left\|P^{(n)}(X^nY^n|\mbf{u})-\tilde{P}^{(n)}(X^nY^n|\mbf{u})\right\|_1<|\mc{U}|\epsilon.
\end{equation}
The size of $\mc{W}_K$ is bounded by
\begin{align}
\log |\mc{W}_K|=\sum_{u\in\mc{U}} \log |\alpha_u|&\leq \sum_{u\in\mc{U}}n_u(I(XY:V|U=u)+\delta)\notag\\
&\leq n(I(XY:V|U)+O(\delta)).
\end{align}

\noindent
\emph{Charlie's Simulation}:
Eq. \eqref{Eq:EndStep1} holds for every $\mbf{u}\in T_{[U]_\delta}^{(n)}$.  The next question is how we choose our code words $\mbf{u}$, which represents the public communication.  If we just took $T^n_{[U]_\delta}$ as the codebook, then the public correlation rate would be $H(U)$.  But we can actually do better, and we will use Wyner's theorem again to construct  a smaller codebook. 

From Wyner, for $n$ sufficiently large there exists a subset $\beta\subset T^n_{[U]_\delta}$ with $|\beta|\leq 2^{n (I(Z:U)+\delta)}$ such that
\[
  \hat{P}^{(n)}(Z^n):=\frac{1}{|\beta|}\sum_{\mbf{u}\in\beta} P^{(n)}(Z^n|\mbf{u})
\]
and
\begin{equation}
\label{Eq:Eve-siumulate}
  \left\| \hat{P}^{(n)}(Z^n)-Q^{(n)}(Z^n) \right\|_1 \leq \epsilon.
\end{equation}
Let $\tilde{U}$ be uniformly distributed over $\beta$ and define
the channel $\tilde{U}|Z^n$ by
\begin{equation}
\label{Eq:Eve_Channel}
\Phi^{(n)}(\mbf{u}|\mbf{z})=\frac{1}{|\beta|}\frac{P^{(n)}(\mbf{z}|\mbf{u})}{\hat{P}^{(n)}(\mbf{z})}\quad\text{for $\mbf{u}\in\beta$}.
\end{equation}
When Charlie applies $\Phi^{(n)}$ to his part of distribution $Q^{(n)}(X^nY^nZ^n)$, the new distribution is given by
\begin{align}
\tilde{Q}^{(n)}(X^nY^nZ^n\tilde{U}):&=Q^{(n)}(X^nY^n|Z^n)\Phi^{(n)}(\tilde{U}|Z^n)Q^{(n)}(Z^n).
\label{eq:tildeQ}
\end{align}
Note that the reduced distribution $\tilde{Q}^{(n)}(X^nY^n\tilde{U})$ is precisely what is obtained when Charlie attempts to simulate the public communication $\tilde{U}$ by acting on $X^nY^nZ^n$ with $\Phi^{(n)}$.  Thus, we want to prove that $\tilde{Q}^{(n)}(X^nY^n\tilde{U})$ is close to the distribution generated by $\tilde{P}^{(n)}(\mbf{x},\mbf{y}|\mbf{u})$ when $\mbf{u}$ is chosen uniformly from $\beta$, which we denote by
\begin{equation*}
\tilde{P}^{(n)}(X^nY^n\tilde{U}):=\frac{1}{|\beta|}\tilde{P}^{(n)}(X^nY^n|\tilde{U}).
\end{equation*}
To do this, we first bound the difference
\begin{align}
&\!\!\!\!\!\!\!\!\!\!\!\!\!\!\!\!\!\!\!\!\!\!\!\!\!\!\!\!\!\!
\left\|\frac{1}{|\beta|}P^{(n)}(X^nY^n|\tilde{U})-\tilde{Q}^{(n)}(X^nY^n\tilde{U})\right\|_1\notag\\
&=\left\|\frac{1}{|\beta|}P^{(n)}(X^nY^n|\tilde{U})-\sum_{\mbf{z}\in\mc{Z}^n}Q^{(n)}(X^nY^n|\mbf{z})\Phi^{(n)}(\tilde{U}|\mbf{z})Q^{(n)}(\mbf{z})\right\|_1\notag\\
&\leq \left\|\frac{1}{|\beta|}P^{(n)}(X^nY^n|\tilde{U})-\sum_{\mbf{z}\in\mc{Z}^n}Q^{(n)}(X^nY^n|\mbf{z})\Phi^{(n)}(\tilde{U}|\mbf{z})\hat{P}^{(n)}(\mbf{z})\right\|_1+\epsilon\notag\\
&=\left\|\frac{1}{|\beta|}P^{(n)}(X^nY^n|\tilde{U})-\frac{1}{|\beta|}\sum_{\mbf{z}\in\mc{Z}^n}P^{(n)}(X^nY^n|\mbf{z})P^{(n)}(\mbf{z}|\tilde{U})\right\|_1+\epsilon\notag\\
&=\left\|\frac{1}{|\beta|}P^{(n)}(X^nY^n|\tilde{U})-\frac{1}{|\beta|}\sum_{\mbf{z}\in\mc{Z}^n}P^{(n)}(X^nY^n|\tilde{U}\mbf{z})P^{(n)}(\mbf{z}|\tilde{U})\right\|_1+\epsilon=\epsilon.
\end{align}
Here, we have used both Eqns. \eqref{Eq:Eve_Channel} and \eqref{Eq:Eve-siumulate}, and the last line follows from the Markov chain condition $XY-Z-U$.  Therefore, combining with Eq. \eqref{Eq:EndStep1}, we obtain the desired result that
\begin{equation}
  \left\| \tilde{Q}^{(n)}(X^nY^n\tilde{U})-\tilde{P}^{(n)}(X^nY^n\tilde{U}) \right\|_1
     \leq \epsilon(1+|\mc{U}|).
\end{equation}

To summarize the protocol, consider any $\delta,\epsilon>0$ and $n$ sufficiently large.  
Either Alice or Bob locally generates the random variable $\tilde{U}$ which is 
uniformly distributed over a set of size $|\beta|\leq 2^{n(I(Z;U)+\delta)}$.  
The value of $\tilde{U}$ is announced publicly.  
Sharing no more than $n(I(XY;V|U)+\delta)$ bits of secret correlation, 
Alice and Bob generate distribution $\tilde{P}^{(n)}(X^nY^n)$ which is 
jointly distributed with $\tilde{U}$ according to $\tilde{P}^{(n)}(X^nY^n\tilde{U})$.  
At the same time, we have shown the existence of a channel $\Phi^{(n)}$ 
such that when Charlie applies this to her part of $X^nY^nZ^n$, it 
generates the distribution $\tilde{Q}^{(n)}(X^nY^n\tilde{U})$ for which
\[
  \left\| \tilde{Q}^{(n)}(X^nY^n\tilde{U})-\tilde{P}^{(n)}(X^nY^n\tilde{U}) \right\|_1
            \leq \epsilon(1+|\mc{U}|).
\]
Therefore, we have satisfied the two components of the achievability criteria.


\end{document}